\documentclass[submission,copyright,creativecommons]{eptcs}
\providecommand{\event}{QPL 2018} 
\usepackage{breakurl}             
\usepackage{underscore}           
\usepackage{import}
\usepackage{latexsym}
\usepackage{amssymb}
\usepackage{amsmath}
\usepackage{amsthm}
\usepackage{stmaryrd}
\usepackage{enumerate}
\usepackage{url}
\usepackage{xypic}
\usepackage{tikz-cd}
\usepackage[T1]{fontenc}
\usepackage{appendix} 

\usepackage{xr}

\title{Categorical Equivalences from State-Effect Adjunctions}
\author{Robert Furber
\institute{Aalborg University\\ Aalborg, Denmark}
\email{furber@cs.aau.dk}
}
\def\titlerunning{Categorical Equivalences from State-Effect Adjunctions}
\def\authorrunning{Robert Furber}

\newif\ifignore

\ignoretrue

\newcommand{\TC}{\ensuremath{\mathcal{TC}}}
\newcommand{\BOUS}{\cat{BOUS}}
\newcommand{\SOUS}{\cat{SOUS}}
\newcommand{\ROUS}{\cat{ROUS}}
\newcommand{\Di}{\ensuremath{\D_\infty}}
\newcommand{\CCL}{\cat{CCL}}
\newcommand{\RBNS}{\cat{RBNS}}
\newcommand{\Rgt}{\ensuremath{\mathbb{R}_{> 0}}}
\newcommand{\Id}{\ensuremath{\mathrm{Id}}}
\newcommand{\BNS}{\cat{BNS}}
\newcommand{\SBNS}{\cat{SBNS}}
\newcommand{\CEMod}{\cat{CEMod}}
\newcommand{\cat}[1]{\ensuremath{\mathbf{#1}}}
\newcommand{\DM}{\ensuremath{\mathcal{DM}}}
\newcommand{\blank}{\ensuremath{{\mbox{-}}}}
\newcommand{\catl}[1]{\ensuremath{\mathcal{#1}}}
\newcommand{\op}{\ensuremath{^{\mathrm{op}}}}
\newcommand{\R}{\ensuremath{\mathbb{R}}}
\newcommand{\Q}{\ensuremath{\mathbb{Q}}}
\newcommand{\Rdn}{\ensuremath{\mathcal{R}}}
\newcommand{\co}{\ensuremath{\mathrm{co}}}
\newcommand{\Z}{\ensuremath{\mathbb{Z}}}
\newcommand{\Tee}{\ensuremath{\mathcal{T}}}
\newcommand{\id}{\ensuremath{\mathrm{id}}}
\newcommand{\D}{\ensuremath{\mathcal{D}}}
\newcommand{\N}{\ensuremath{\mathbb{N}}}
\newcommand{\Rgeq}{\ensuremath{\mathbb{R}_{\geq 0}}}
\newcommand{\OUS}{\cat{OUS}}
\newcommand{\Eff}{\ensuremath{\mathcal{F}}}
\newcommand{\Hil}{\ensuremath{\mathcal{H}}}
\newcommand{\Proj}{\ensuremath{\mathrm{Proj}}}
\newcommand{\BAff}{\ensuremath{\mathrm{BAff}}}
\newcommand{\EA}{\cat{EA}}
\newcommand{\Stat}{\ensuremath{\mathrm{Stat}}}
\newcommand{\CStat}{\mathrm{CStat}}
\newcommand{\BBNS}{\cat{BBNS}}
\newcommand{\EM}{\ensuremath{\mathcal{E\!M}}}
\newcommand{\Smith}{\cat{Smith}}
\newcommand{\Norm}{\cat{Normed}}
\newcommand{\CHaus}{\cat{CHaus}}
\newcommand{\EMod}{\cat{EMod}}
\newcommand{\Set}{\cat{Set}}
\newcommand{\BoolAlg}{\cat{BA}}
\newcommand{\FinSet}{\cat{FinSet}}
\newcommand{\FinBA}{\cat{FinBA}}
\newcommand{\Stone}{\cat{Stone}}

\newtheorem{theorem}{Theorem}[section]
\newtheorem{lemma}[theorem]{Lemma}
\newtheorem{proposition}[theorem]{Proposition}
\newtheorem{example}[theorem]{Example}

    { \begin{trivlist}%
        \item[\hskip \labelsep {\bfseries #1}]%
    }%
    { \end{trivlist}%
    }

\newcommand{\fin}{\mathrm{fin}}
\renewcommand{\Eff}{\mathcal{E}}
\renewcommand{\Stat}{\mathcal{S}}
\newcommand{\EffS}{\Eff_{\pm}}
\newcommand{\StatS}{\Stat_{\pm}}
\newcommand{\CEff}{\mathcal{CE}}
\renewcommand{\CStat}{\mathcal{CS}}
\newcommand{\CEffS}{\mathcal{CE}_{\pm}}

\begin{document}
\maketitle

\begin{abstract}
From every pair of adjoint functors it is possible to produce a (possibly trivial) equivalence of categories by restricting to the subcategories where the unit and counit are isomorphisms. If we do this for the adjunction between effect algebras and abstract convex sets, we get the surprising result that the equivalent subcategories consist of reflexive order-unit spaces and reflexive base-norm spaces, respectively. These are the convex sets that can occur as state spaces in generalized probabilistic theories satisfying both the no-restriction hypothesis and its dual. The linearity of the morphisms is automatic. If we add a compact topology to either the states or the effects, we can obtain a duality for all Banach order-unit spaces or all Banach base-norm spaces, but not both at the same time.
\end{abstract}

\section{Introduction}
Effect algebras, first defined in \cite{FoulisB94}, are an algebraic structure that puts a common roof over Boolean algebras, orthomodular lattices, and unit intervals of (unital) C$^*$-algebras. Every effect algebra has a state space, which is a convex set. Another approach is to define the state space first, as is done in \cite{barrett07} or to define states and effects simultaneously, as in \cite{barnum2010}. 

Bart Jacobs \cite[Theorem 17]{Jacobs10e} showed that states and effects form a dual adjunction between the category of effect algebras and the category of Eilenberg-Moore algebras of the distribution monad, for which the real unit interval $[0,1]$ is a dualizing object. There is also an earlier, less categorical approach by \cite{beltrametti}, where countable additivity was required for states.  For any adjunction, we can restrict the functors involved to the subcategories where the unit and counit are isomorphisms to obtain a duality (possibly between two empty categories). This produces equivalent categories of effect algebras that are determined by their states and convex state spaces that are determined by their effects. From the point of view of generalized probabilistic theories, the convex sets occurring are those convex sets $X$ that can inhabit generalized probabilistic theories having the \emph{no-restriction hypothesis} on effects \cite[II.D Definition 16]{chiribella2010} as well as the dual of this, \emph{i.e.} all states on the algebra of effects already occur in $X$. 

In order to provide a characterization, we use a relationship of effect algebras and convex sets to an older theory of operational quantum mechanics in terms of base-norm and order-unit spaces \cite{edwards, davies, barnum}. Using nontrivial fullness results, we characterize the subcategories where the unit and counit are isomorphisms as the effect algebras that are unit intervals in reflexive order-unit spaces and the abstract convex sets arising from the bases of reflexive base-norm spaces. We then briefly describe how to modify the state and effect functors by adding topologies to obtain two dual adjunctions, one having Banach base-norm spaces and Smith order-unit spaces as the resulting duality, and the other having Smith base-norm spaces and Banach order-unit spaces as the resulting duality. 

In the appendix we show what form of the archimedean property is necessary for the automatic $\R$-linearity result we use, by means of counterexamples. 

To the reader familiar with Stone duality, we offer the following analogy. A very ``bare-bones'' form of Stone duality can be obtained between $\Set$ and $\BoolAlg$, the category of Boolean algebras, by using $2$, which is both a set and a Boolean algebra, as a dualizing object, with $\Set(\blank,2) : \Set \rightarrow \BoolAlg\op$ being isomorphic to the contravariant powerset functor and $\BoolAlg(\blank,2)$ being isomorphic to the functor taking a Boolean algebra to its set of ultrafilters. The duality is only an adjunction, where $\Set(\blank,2)$ is a left adjoint to $\BoolAlg(\blank, 2)$. 

Restricting this to the subcategories where the unit and counit are isomorphisms produces a duality between $\FinSet$ and $\FinBA\op$, the finite version of Stone duality. In order to work with all Boolean algebras, Stone equipped $X = \BoolAlg(A,2)$ with a compact Hausdorff topology (the Stone space) and restricted the maps $X \rightarrow 2$ to be continuous, producing an equivalence $\Stone \simeq \BoolAlg\op$ \cite[\S 4.4 Corollary]{johnstone}. If instead we want to use all sets, we can put a compact Hausdorff topology on $A = \Set(X,2)$ and require maps $A \rightarrow 2$ to be continuous. This is the duality between complete atomic Boolean algebras and sets \cite[VI \S 3.1]{johnstone}. The dualities in Section \ref{SmithSpaceSection} are analogous to these dualities, but with effect algebras instead of Boolean algebras, and $[0,1]$ instead of $2$. 

\section{Preliminaries}
\label{OrigAdjSection}
In this section we go through the definitions of the various objects involved. Effect algebras, first defined in \cite{FoulisB94}, are an algebraic structure that puts a common roof over Boolean algebras, orthomodular lattices, and unit intervals of C$^*$-algebras. A standard reference is \cite{DvurecenskijP00}, where the definition is given in \cite[Definition 1.2.1]{DvurecenskijP00}. They are defined as follows. A partial commutative monoid is a triple $(A,0,\ovee)$ where $A$ is a set, $0 \in A$ an element, $\ovee : A \times A \rightarrow A$ a partial function, subject to some axioms we shall specify in a moment. We use the notation $a \perp b$ for $a, b \in A$ to mean that $(a,b)$ is in the domain of $\ovee$, \emph{i.e.} that $a \ovee b$ is defined. The axioms for a partial commutative monoid are:
\begin{enumerate}[(i)]
\item Identity axiom: For all $a \in A$, $0 \perp a$ and $0 \ovee a = a$. 
\item Associativity axiom: For all $a,b,c \in A$, if $b \perp c$ and $a \perp (b \ovee c)$, then $a \perp b$, $(a \ovee b) \perp c$, and $a \ovee (b \ovee c) = (a \ovee b) \ovee c$. 
\item Commutativity axiom: For all $a,b \in A$, if $a \perp b$, then $b \perp a$ and $a \ovee b = b \ovee a$. 
\end{enumerate}
In the absence of the commutativity axiom, the identity and associativity axioms would only be half-finished, but the reader may verify that they also hold the other way using commutativity. An \emph{effect algebra} is a quintuple $(A,0,\ovee,\blank^\bot, 1)$ such that $(A,0,\ovee)$ is a partial commutative monoid, $\blank^\bot : A \rightarrow A$ is a (total) function, and $1 \in A$ an element such that
\begin{enumerate}[(i)]
\item $1 = 0^\bot$.
\item For all $a \in A$, $a^\bot \ovee a = 1$ and $a^\bot$ is the unique such element of $A$. 
\item If $a \in A$ and $a \perp 1$, then $a = 0$. 
\end{enumerate}
The operation $\blank^\bot$ is known as the orthosupplement. An orthomodular lattice $(A, 0, 1, \land, \lor, \lnot)$ can be given the structure of an effect algebra by defining $a \perp b$ iff $a \leq \lnot b$ (equivalently, $b \leq \lnot a$) and $a \ovee b = a \lor b$ in this case, keeping $0 = 0$ and $1 = 1$, and defining $a^\perp = \lnot a$. As Boolean algebras are (distributive) orthomodular lattices, we have that every Boolean algebra can be made into an effect algebra in the same manner. The interval $[0,1]$ can be made into an effect algebra by taking $0=0$, $1=1$, $a \perp b$ iff $a + b \leq 1$ and then $a \ovee b = a + b$, and $a^\bot = 1 - a$. 

Morphisms of effect algebras are defined to be maps $A \rightarrow B$ such that $f(1) = 1$ and if $a \perp a'$ in $A$, then $f(a) \perp f(a')$ in $B$ and $f(a \ovee a') = f(a) \ovee f(a')$. The reader may verify that these rules imply $f(a^\bot) = f(a)^\bot$, so in particular, $f(0) = 0$. Effect algebras and their morphisms form a category $\EA$. To aid the reader in understanding the usefulness of making the addition partial, observe that, for any Boolean algebra $A$, the hom-set $\EA(A,[0,1])$ is exactly the set of finitely additive probability measures on $A$. Similarly, if $B$ is an orthomodular lattice, $\EA(B,[0,1])$ is the set of (finitely-additive) states.

We now move on to convex sets. Convex sets, in the abstract version that we use here, are intended to be used to represent mixed state spaces of physical systems. The idea is that if $X$ is a set of states, and $x,y \in X$ are states, we can form the state that is $x$ with probability $p \in [0,1]$ and $y$ with probability $(1-p) \in [0,1]$. This operation defines a mapping $\kappa : [0,1] \times X \times X \rightarrow X$, which is then required to satisfy some natural axioms. Details of how to do this can be found in \cite{Stone1949, Neumann1970, Gudder1973, ozawa80, Fritz09a}. In this article, however, we prefer to use the approach based on monads. Let $\D(X)$ be the set of functions $\phi : X \rightarrow [0,1]$ with finite support, such that $\sum_{x \in X} \phi(x) = 1$. We can extend $\D$ to a functor by defining, for each $f : X \rightarrow Y$ the function $\D(f) : \D(X) \rightarrow \D(Y)$ to be
\[
\D(f)(\phi)(y) = \sum_{x \in f^{-1}(y)} \phi(x).
\]
We can then define the natural transformations $\eta$ and $\mu$ as follows. For $\eta_X : X \rightarrow \D(X)$ we define\footnote{We use $[P]$ for the Iverson bracket, \emph{i.e.} to mean a number that is $1$ if $P$ is true and $0$ if $P$ is false.} $\eta_X(x)(x') = [x = x']$, and for $\mu_X : \D(\D(X)) \rightarrow \D(X)$ we define
\[
\mu_X(\Phi)(x) = \sum_{\phi \in \D(X)} \Phi(\phi) \cdot \phi(x).
\]
With these definitions, $(\D,\eta,\mu)$ is a monad. 

A convex set will be defined to be an \emph{Eilenberg-Moore algebra}\footnote{See \cite[VI.2]{maclane} for definitions and basic theorems.} of $\D$, which is to say, a pair $(X,\alpha : \D(X) \rightarrow X)$ such that
\begin{equation}
\label{DistEADiag}
\begin{tikzcd}
X \ar{r}{\eta_X} \ar{rd}[swap]{\id_X} & \D(X) \ar{d}{\alpha_X} & \D(\D(X)) \ar{r}{\D(\alpha_X)} \ar{d}[swap]{\mu_X} & \D(X) \ar{d}{\alpha_X} \\
 & X & \D(X) \ar{r}[swap]{\alpha_X} & X,
\end{tikzcd}
\end{equation}

Elements of $\D(X)$ have an interpretation as instructions for how to produce a mixed state. For example, we can define a function $\phi : [0,1] \rightarrow [0,1]$ such that $\phi(0) = \frac{1}{2}$, $\phi(\frac{1}{3}) = \frac{1}{2}$, and $\phi(\alpha) = 0$ for all other $\alpha \in [0,1]$. This $\phi$ tells us to use the state $0 \in [0,1]$ with probability $\frac{1}{2}$ and the state $\frac{1}{3}$ with probability $\frac{1}{2}$. The mapping $\alpha : \D(X) \rightarrow X$ implements these instructions, mapping specifications for how to mix states to the actual mixtures. A simple example is $[0,1]$:
\[
\alpha_{[0,1]}(\phi) = \sum_{x \in [0,1]} \phi(x) \cdot x.
\]
Using the example $\phi$ above, we get $\alpha(\phi) = \frac{1}{6}$, as things should be. The physical interpretation of the left hand diagram of \eqref{DistEADiag} is that mixing a state only with itself produces the same state again, and the physical interpretation of the right hand diagram is a kind of associativity -- if we start with a specification for how to mix specifications of mixed states, we get the same result if we form the mixed states at the lower level, and then mix the results (the top right path) or we mix the specifications in the natural way using $\mu_X$ and then form the mixed state using the resulting specification.

One benefit of this method of defining convex sets is that we can change $\D$ to some other monad and get a new kind of convex sets. For instance, if $\Di(X)$ is defined to be functions $\phi : X \rightarrow [0,1]$ that are summable, in the sense that $\sum_{x \in X}\phi(x)$ converges, then $\EM(\Di)$ is the category of $\sigma$-convex sets. It appears that \'Swirszcz was the first to describe convex sets using monads \cite[\S 4]{swirszcz75}. 

We can consider an effect to be a ``fuzzy predicate'', \emph{i.e.} a map from a convex set $X$ to $\D(2)$, where $2$ is a set with two elements. In particular, we can identify these two elements with $0, 1 \in \R$ and $\D(2)$ with $[0,1]$, their convex hull. Therefore the effects functor $\Eff : \EM(\D) \rightarrow \EA\op$ is defined to be the hom-functor $\EM(\D)(\blank,[0,1])$ on underlying sets, and with the effect algebra structure taken to be that of $[0,1]$, defined pointwise \cite[\S 6, Lemma 16]{Jacobs10e}. 

We can consider every mixed state of a physical system described by an effect algebra $A$ to define a map $\phi : A \rightarrow [0,1]$, which maps each effect to its probability of being true when measured. Such a map should take $0$ to $0$ and $1$ to $1$, and be additive on disjoint elements, \emph{i.e.} if $a,b \in A$ are such that $a \ovee b$ exists, then $\phi(a) + \phi(b) \in [0,1]$, and is equal to $\phi(a \ovee b)$. This is equivalent to the map $\phi : A \rightarrow [0,1]$ being a morphism of effect algebras, and so this is the usual definition of a \emph{state} on an effect algebra \cite[Definition 1.3.3]{DvurecenskijP00}. The state functor $\Stat : \EA\op \rightarrow \EM(\D)$ is defined to be the hom functor $\EA(\blank, [0,1])$ on underlying sets, and to use the $\EM(\D)$-structure of $[0,1]$ pointwise to define the $\EM(\D)$-structure of $\EA(A,[0,1])$ \cite[\S 6, Lemma 15]{Jacobs10e}.

Every element $a$ of an effect algebra $A$ defines an effect on its state space, which defines a natural mapping
\begin{align*}
\epsilon_A &: A \rightarrow \Eff(\Stat(A)) \\
\epsilon_A(a)(\phi) &= \phi(a),
\end{align*}
where $\phi \in \Stat(A)$.

Similarly, every element $x$ of an abstract convex set $X$ defines a state of $\Eff(X)$, which defines a natural mapping
\begin{align*}
\eta_X &: X \rightarrow \Stat(\Eff(X)) \\
\eta_X(x)(a) &= a(x),
\end{align*}
where $a \in \Eff(X)$.

If we use the convention that $\Eff : \EM(\D) \rightarrow \EA\op$ and $\Stat : \EA\op \rightarrow \EM(\D)$, then the mappings defined above make $\Eff$ a left adjoint to $\Stat$, \emph{i.e.} the following diagrams commute

\begin{tikzcd}
\Stat(A) \ar{r}{\eta_{\Stat(A)}} \ar{rd}[swap]{\id} & \Stat(\Eff(\Stat(A))) \ar{d}{\Stat(\epsilon_A)} & \Eff(X) & \Eff(\Stat(\Eff(X))) \ar{l}[swap]{\Eff(\eta_X)} \\
 & \Stat(A) & & \Eff(X) \ar{ul}{\id} \ar{u}[swap]{\epsilon_{\Eff(X)}}.
\end{tikzcd}

See \cite[\S 6, Theorem 17]{Jacobs10e} for a proof of this. The reason the arrows are reversed for $\epsilon$ and the second diagram is our choice of $\EA\op$ when turning contravariant functors into covariant ones. The equivalence of the above definition of adjunction to the other definitions of adjunction is proven in \cite[IV.1 Theorem 2]{maclane}.

If we have a pair of adjoint functors $F : \catl{D} \rightarrow \catl{C}$ and $G : \catl{C} \rightarrow \catl{D}$ with unit $\eta : \Id \Rightarrow GF$ and counit $\epsilon : FG \Rightarrow \Id$, we can always define subcategories
\begin{align*}
\catl{C}' &= \{ A \in \catl{C} \mid \epsilon_A \text{ is an isomorphism} \} \\
\catl{D}' &= \{ X \in \catl{D} \mid \eta_X \text{ is an isomorphism} \},
\end{align*}
where the morphisms are just those of $\catl{C}$ and $\catl{D}$ that are between objects in $\catl{C}'$ and $\catl{D}'$ respectively (\emph{i.e.} $\catl{C}'$ and $\catl{D}'$ are \emph{full} subcategories).
Then $F' : \catl{D}' \rightarrow \catl{C}'$ and $G' : \catl{C}' \rightarrow \catl{D}'$, the restrictions of $F$ and $G$, form an equivalence of categories \cite[Part 0, Proposition 4.2]{lambek}. It is natural to ask what $\catl{C}'$ and $\catl{D}'$ are for $\Eff$ and $\Stat$, as this answers the question of which effect algebras are determined by their states and which abstract convex sets are determined by their effects. The purpose of the rest of the article is to answer this question. 

\section{Base-Norm Spaces, Order-Unit Spaces and Full Embeddings}
\label{BNSOUSSection}
Base-norm spaces and order-unit spaces are particular kinds of ordered vector space designed to characterize state spaces and spaces of effects. For background in terms of operational theories in physics, see \cite{edwards,davies,barnum}. Some standard mathematical references are \cite{asimow, alfsen71, alfsen01}. 

A \emph{cone} in a real vector space $E$ is a set $E_+ \subseteq E$ such that if $x,y \in E_+$, then $x + y \in E_+$ and if $\alpha \in \Rgeq$ and $x \in E_+$, then $\alpha x \in E_+$, and $E_+ \cap -E_+ = \{ 0 \}$. Cones are a convenient way of defining orderings on the elements of a vector space. We define $x \leq y$ iff $y - x \in E_+$. Under this definition, $\leq$ is a partial order, $x \geq 0$ iff $x \in E_+$, and the expected implications $x \leq y \Leftrightarrow x + z \leq y + z$ and $x \leq y \Rightarrow \alpha x \leq \alpha y$ for $\alpha \in \Rgeq$ hold. We define an \emph{ordered vector space} to be a pair $(E, E_+)$ where $E$ is a real vector space and $E_+$ is a cone. The order is directed iff the cone is \emph{generating}, \emph{i.e.} $E = E_+ - E_+$. Linear maps $f : (E,E_+) \rightarrow (F,F_+)$ are called positive if $f(E_+) \subseteq F_+$. For linear maps, this is equivalent to monotonicity, \emph{i.e.} if $x \leq y$ then $f(x) \leq f(y)$. If we take $\Rgeq \subseteq \R$ as a cone, we get the usual ordering on $\R$, making $(\R,\Rgeq)$ an ordered vector space. 

A convex set can be embedded in many different vector spaces. The purpose of base-norm spaces is to provide a particularly convenient form for certain convex sets to live in. 

If we have a triple $(E,E_+,\tau)$, where $(E,E_+)$ is a directed ordered vector space, and $\tau : E \rightarrow \R$ is a positive linear map, which, if $E \neq \{0\}$, takes a non-zero value on some element of $E$, then we define the \emph{base}
\[
B(E) = \{ x \in E_+ \mid \tau(x) = 1 \} = \tau^{-1}(1) \cap E_+,
\]
and the \emph{unit ball} $U(E)$ to be the absolutely convex hull\footnote{The absolutely convex hull of a set $X$ is the smallest subset containing $X$ closed under absolutely convex combinations, which are like convex combinations, but using numbers $\alpha_i \in \R$ such that $\sum_{i=1}^n |\alpha_i| \leq 1$ rather than $\alpha_i \in [0,1]$ with $\sum_{i=1}^n \alpha_i = 1$.} (or circled convex hull) of the base. If $B(E)$ is not empty, $U(E) = \co(B(E) \cup -B(E))$, and this is taken as the definition by many authors because they do not consider the case of an empty $B(E)$. If $U(E)$ is \emph{absorbent}\footnote{or \emph{absorbing}}, \emph{i.e.} for each $x \in E$, there exists $\alpha \in \Rgeq$ such that $x \in \alpha U(E)$, then the \emph{Minkowski functional}\footnote{or \emph{gauge} \cite[II.1.4, page 39]{schaefer}} $\|\blank\|_{U(E)}$, defined by
\[
\|x\|_{U(E)} = \inf\{\alpha \in \Rgt \mid x \in \alpha U(E) \}
\]
is a seminorm. We say that $U(E)$ is \emph{radially bounded} (respectively, \emph{radially compact}) if for each element $x \in U(E)$, the set $\{ \alpha \in \R \mid \alpha x \in U(E) \}$ is bounded (respectively, compact) in $\R$. The Minkowski functional is a norm iff $U(E)$ is radially bounded. 

A \emph{base-norm space} is defined to be a triple $(E,E_+,\tau)$, where $(E,E_+)$ is a directed ordered vector space, $\tau$ is a positive linear map $E \rightarrow \R$ that takes a non-zero value if $E \neq \{0\}$, such that $U(E)$ is radially bounded, and $E_+$ is closed with respect to the norm $\|\blank\|_{U(E)}$. If $U(E)$ is radially compact, $E_+$ is automatically closed \cite[Proposition 2.2.6 (ii)]{furberthesis}. Some authors take radial compactness as part of the definition \cite[Definition 1.10]{alfsen01}, though it is not equivalent to the definition used here \cite[Counterexample A.6.2]{furberthesis}.

Though it might not be immediately apparent, $\tau$ is strictly positive, \emph{i.e.} $\tau(x) = 0$ implies $x = 0$ \cite[Lemma 2.2.2]{furberthesis}, and $\|\blank\|_{U(E)}$ is a norm \cite[Lemma 2.2.3]{furberthesis}. For conciseness, we will use $\|\blank\|$ instead of $\|\blank\|_{U(E)}$ from now on, using $\|\blank\|_E$ if we need to specify the space. A \emph{Banach} base-norm space is one that is complete with respect to its norm. 

The map $\tau$ is called the \emph{trace}. There are equivalent definitions of base-norm space where the base $B(E)$ is part of the definition and $\tau$ is derived from it. There are also inequivalent definitions in use, which have their own reasons for existing, but are not suited to duality between states and effects \cite[\S 2.2.3]{furberthesis}.

Base-norm spaces form a category $\BNS$ with linear, positive, trace-preserving maps, \emph{i.e.} linear positive maps $f : (E,E_+,\tau) \rightarrow (F,F_+,\sigma)$ such that $\sigma \circ f = \tau$. Banach base-norm spaces form a full subcategory $\BBNS$. 

It may help to give an example of a base-norm space. If we take $E = \R^X_\fin$ to be the direct sum of $\R$, using the index set $X$ (\emph{i.e.} $E$ is the finitely supported functions $X \rightarrow \R$), and define $E_+ = \{ \phi \in E \mid \forall x \in X . \phi(x) \geq 0 \}$, and $\tau(\phi) = \sum_{x \in X}\phi(x)$, then $(E,E_+,\tau)$ is a base-norm space with base $\D(X)$. Furthermore, its completion is the space of absolutely summable sequences $\ell^1(X)$, which is a base-norm space with base $\Di(X)$. 

One of the examples that motivated the definition of base-norm space was the space of self-adjoint trace-class operators $\TC(\Hil)$ on a Hilbert space $\Hil$. The positive cone is the set of positive (semidefinite) operators, and the trace is the usual notion of trace. The base of $(\TC(\Hil), \TC(\Hil)_+, \tau)$ is the convex set of density matrices, $\DM(\Hil)$ \cite{edwards} \cite[\S 2.3]{furberthesis}. In fact, this example and the previous example of $\ell^1(X)$ are both examples of self-adjoint parts of the preduals of W$^*$-algebras, the W$^*$-algebras being $B(\Hil)$ and $\ell^\infty(X)$ respectively. 

The base $B(E)$ of a base-norm space $(E,E_+,\tau)$ is a convex subset of $E$, and so it has an $\EM(\D)$-structure $\alpha_{B(E)} : \D(B(E)) \rightarrow B(E)$ defined by
\[
\alpha_{B(E)}(\phi) = \sum_{x \in B(E)} \phi(x) \cdot x,
\]
which is to say, we interpret \emph{formal} convex combinations as \emph{actual} convex combinations. As trace-preserving maps map bases to bases, $B$ is a functor $B : \BNS \rightarrow \EM(\D)$. In fact, in Banach base-norm spaces, the base has $\sigma$-convex combinations, so $B$ can be considered to be a functor $B : \BBNS \rightarrow \EM(\Di)$. 

The following is an extension of the result \cite[Theorem 3.4]{Gudder1973} from endomorphisms to arbitrary morphisms. It is also related to \cite[Theorem 3.3]{pumplun02}, but that paper uses a different definition of base-norm space (the positive cone is not required to be closed). A version of the following already appeared in \cite[Proposition 2.4.8]{furberthesis}. 
\begin{proposition}
\label{ConvFullnessProp}
The functor $B : \BNS \rightarrow \EM(\D)$ is full and faithful.
\end{proposition}
\begin{proof}
As it is easier, we prove that it is faithful first. Let $f,g : E \rightarrow F$ in $\BNS$. If $E = \{ 0 \}$, then $f = g$ by linearity, so we may assume that $E \neq \{ 0 \}$. Then the linear span of $B(E)$ is $E$, so $B(f) = B(g)$, \emph{i.e.} $f$ agrees with $g$ on $B(E)$, implies $f = g$, by linearity again. 

To show that $B$ is full, let $f : B(E) \rightarrow B(F)$ be an $\EM(\D)$-map, \emph{i.e.} an affine map, where $(E,E_+,\tau)$ and $(F,F_+,\sigma)$ are base-norm spaces. If $E = \{ 0 \}$, then $B(E) = \emptyset$ and $f$ is the empty function. Taking $g : \{0\} \rightarrow F$ to be the unique linear map, it is positive and trace preserving and $B(g) = f$. So we now assume that $E \neq \{0\}$. 

Every element $x \in E$ can be expressed as $\alpha x_+ - \beta x_-$, with $\alpha,\beta \in \Rgeq$ and $x_+,x_- \in B(E)$. We define
\[
g(\alpha x_+ - \beta x_-) = \alpha f(x_+) - \beta f(x_-).
\]
We must first show that this is well-defined. Suppose that $\alpha x_+ - \beta x_- = \alpha' x_+' - \beta' x_-'$. Then $\alpha x_+ + \beta' x_-' = \alpha'x_+' + \beta x_-$, and taking the trace, $\alpha + \beta' = \alpha' + \beta$. Define $\gamma = \alpha + \beta'$. If $\gamma = 0$, then $\alpha = -\beta'$, so both are zero because they are nonnegative. Likewise $\alpha' = \beta = 0$. Therefore $g(\alpha x_+ - \beta x_-) = 0 = g(\alpha'x_+' - \beta'x_-')$. 

If, on the other hand, $\gamma \neq 0$, then the convex combinations $\frac{\alpha}{\gamma}x_+ + \frac{\beta'}{\gamma}x_-'$ and $\frac{\alpha'}{\gamma}x_+' + \frac{\beta}{\gamma}x_-$ exist and define equal elements of $B(E)$. As $f$ is an affine map, we have
\[
\frac{\alpha}{\gamma}f(x_+) + \frac{\beta'}{\gamma}f(x_-') = f\left(\frac{\alpha}{\gamma}x_+ + \frac{\beta'}{\gamma}x_-'\right) = f\left(\frac{\alpha'}{\gamma}x'_+ + \frac{\beta}{\gamma}x_-\right) = \frac{\alpha'}{\gamma}f(x_+') + \frac{\beta}{\gamma}f(x_-).
\]
If we then multiply through by $\gamma$, and subtract from both sides, we can conclude
\[
g(\alpha x_+ - \beta x_-) = \alpha f(x_+) - \beta f(x_-) = \alpha'f(x_+') - \beta' f(x_-') = g(\alpha' x_+' - \beta' x_-').
\]
The map $g$ can be proved to preserve addition by decomposing the elements involved as $\alpha x_+ - \beta x_-$, and using a scale factor to use the fact that $f$ is an affine map. The proofs that $g$ preserves scalar multiplication, positivity, and the trace are simple applications of its definition, as is the proof that $B(g) = f$. Therefore $B$ is full. 
\end{proof}

A consequence of the previous result is that all affine morphisms between bases of Banach base-norm spaces are also $\sigma$-affine. 

We now move on to order-unit spaces. These are vector spaces that are to effect algebras what base-norm spaces are to convex sets. If $(A,A_+)$ is an ordered vector space, then an element $u \in A_+$ is said to be a \emph{strong order unit} if for all $a \in A$, there exists an $\alpha \in \Rgeq$ such that $-\alpha u \leq a \leq \alpha u$. Note that if an ordered vector space has a strong order unit, it must be directed. We say that $(A,A_+,u)$, where $u$ is a strong order unit, is \emph{archimedean} if $a \leq \frac{1}{n} u$ for all $n \in \N$ implies that $a \leq 0$ (\emph{i.e.} $x \in -A_+$). Be warned that this is strictly stronger than the condition that one might naturally call archimedeanness, that if $0 \leq a \leq \frac{1}{n} u$, then $a = 0$. 

An \emph{order-unit space} is a triple $(A,A_+,u)$ such that $(A,A_+)$ is an ordered vector space, $u$ is an archimedean strong order unit. We can define the unit interval $[0,u] = \{a \in A \mid 0 \leq a \leq u \}$, and the unit ball $[-u,u] = \{a \in A \mid -u \leq a \leq u \}$. The unit ball is a radially bounded absolutely convex set, so the Minkowski functional $\|\blank\|_{[-u,u]}$ is a norm. If an order-unit space is complete in this norm, we say that it is a \emph{Banach order-unit space}. Typical examples of Banach order-unit spaces are the spaces of self-adjoint elements of unital C$^*$-algebras, for example the self-adjoint part of $B(\Hil)$ for a Hilbert space $\Hil$ or the space of real-valued continuous functions $C(X)$ on a compact Hausdorff space $X$. 

Order-unit spaces form a category $\OUS$ with linear, positive maps that preserve the unit element. Banach order-unit spaces are a full subcategory $\BOUS$. 

For each order-unit space $(A,A_+,u)$, we can equip its unit interval $[0,1]_A$ with the structure of an effect algebra as follows. We take $1$ to be the unit $u$, $0$ to be the element $0$, $a^\perp = u - a$, and define $a \ovee b$ iff $a + b \leq u$, in which case it is defined to be $a + b$. If applied to the order-unit space $(\R,\Rgeq,1)$, this produces the usual effect algebra structure on $[0,1]$. As positive unital maps of order-unit spaces preserve the unit interval, this defines a functor $[0,1]_\blank : \OUS \rightarrow \EA$. Such effect algebras are modules over $[0,1]$, and so are special cases of \emph{convex effect algebras} defined in \cite{PulmannovaG98}, and called effect modules in \cite[\S 3]{JacobsM12b}, where they given a different definition by treating $[0,1]$ as in internal monoid in $\EA$.

\begin{proposition}
\label{OUSFullnessProp}
The functor $[0,1]_\blank : \OUS \rightarrow \EA$ is full and faithful.
\end{proposition}
\begin{proof}
Again, it is simpler to prove that it is faithful first. Let $f,g : (A,A_+,u) \rightarrow (B, B_+, v)$ be positive unital maps of order-unit spaces such that $[0,1]_f = [0,1]_g$, \emph{i.e.} $f$ and $g$ agree on $[0,1]_A$. The fact that $u$ is a strong order unit implies that $A$ is the linear span of $[0,1]_A$, and so the linearity of $f$ and $g$ implies that they are equal. 

We now prove that $[0,1]_\blank$ is full. So let $f : [0,1]_A \rightarrow [0,1]_B$ be an effect algebra homomorphism. We proceed in steps, first showing that it extends to a monoid homomorphism $g : A_+ \rightarrow B_+$, then to a positive unital group homomorphism $h : A \rightarrow B$, then that this is necessarily a $\Q$-linear map, that it is continuous, and therefore that it is in fact $\R$-linear, so a positive linear map of order-unit spaces. 

We define $g$ as follows. Since $u$ is a unit, we have that for all $x \in A_+$ there is an $n \in \N$ such that $x \leq nu$, and therefore $\frac{1}{n}x \in [0, u] = X$. Then $g(x) = n f\left(\frac{1}{n}x\right)$.
First we must show that this is well-defined. Let $m,n \in \N$ such that $\frac{1}{m}x, \frac{1}{n}x \in [0,u]$. Then
\[
n f\left(\frac{1}{n}x\right) = n f\left(m \cdot \frac{1}{mn} x \right) = n m f\left(\frac{1}{mn}x \right) = m f \left(n \cdot \frac{1}{mn} x \right) = m f \left( \frac{1}{m} x \right),
\]
which establishes well-definedness.

Now, $g(0) = f(1 \cdot 0) = f(0) = 0$ because $f$ is a morphism of effect algebras. Now, consider $a, b \in A_+$. There is some $n \in \N$ such that $\frac{1}{n}(a + b) \in [0, u]$. We have
\[
g(a+b) = n f\left(\frac{1}{n}(a + b) \right) = n f\left(\frac{1}{n}a + \frac{1}{n}b \right)
\]
By monotonicity of $+$, $\frac{1}{n}a$ and $\frac{1}{n}b$ are in $[0,u]$ too, and the sum is an effect algebra sum, so since $f$ is an effect algebra homomorphism, we can continue
\[
 = n \left( f\left(\frac{1}{n}a\right) + f\left(\frac{1}{n}b\right)\right) = n f\left(\frac{1}{n}a\right) + n f\left(\frac{1}{n}b\right) = g(a) + g(b).
\]
This concludes the proof that $g$ is a monoid homomorphism. 

We can then define the extension to a group homomorphism $h : A \rightarrow B$, which is necessarily positive and unital because it agrees with $g$ on $A_+$ and with $f$ on $[0,1]_A$, as follows. Recall that each element $a$ of $A$ can be expressed as $a_+ - a_-$, where $a_+, a_- \in A_+$.  We define $h(a) = g(a_+) - g(a_-)$. The proof that this is well-defined is similar to, but simpler than, the proof of the well-definedness of $g$ in Proposition \ref{ConvFullnessProp}, so is omitted. We have $h(0) = h(0 - 0) = g(0) - g(0) = 0$ because $g$ is a monoid homomorphism. Now, if $a, b \in A$,
\begin{align*}
h(a + b) &= h((a_+ - a_-) + (b_+ - b_-)) = h((a_+ + b_+) - (a_- + b_-)) \\
 &= g(a_+ + b_+) - g(a_- + b_-) = g(a_+) + g(b_+) - g(a_-) - g(b_-) \\
 &= g(a_+) - g(a_-) + g(b_+) - g(b_-) = h(a_+ - a_-) + h(b_+ - b_-) \\
 &= h(a) + h(b),
\end{align*}
showing that $h$ is a group homomorphism.

We now show that $h$ preserves multiplication by rational numbers, and so is $\Q$-linear. Let us consider the case of $\frac{1}{m}$ for $m \in \N$ first. We first observe that
\[
m \cdot h\left(\frac{1}{m} a \right) = \sum_{i = 1}^m h \left(\frac{1}{m} a \right) = h \left( \sum_{i = 1}^m \frac{1}{m} a \right) = h(a).
\]
We can conclude from this that $h\left(\frac{1}{m}a\right) = \frac{1}{m}h(a)$. In fact, since $h$ is a group homomorphism, it preserves negation and so this is true for $m \in \Z \setminus \{0\}$. If we now consider the case of a general rational number $\frac{n}{m} \in \Q$, taking the denominator to be negative if the number is negative, we have
\[
f\left(\frac{n}{m}a\right) = f \left( \sum_{i = 1}^n \frac{1}{m} a \right) = \sum_{i =1}^n f\left(\frac{1}{m} a \right) = n \cdot \frac{1}{m} f(a) = \frac{n}{m} f(a).
\]

We now establish the continuity of $h$. It is helpful at this point to recall some useful facts. For a normed space $E$ with closed unit ball $U$, $\{\frac{1}{n}U\}_{n \in \N}$ is a neighbourhood basis for $0$. A group homomorphism between two topological groups is continuous iff it is continuous at the identity element \cite[\S III.8, Proposition 23]{bourbaki}. So we show $h$ is continuous by showing that it is continuous at $0$, and we can do this by showing that for every $n \in \N$, there is an $m \in \N$ such that
\[
h\left(\left[-\frac{1}{m}u,\frac{1}{m}u\right]\right) \subseteq \left[-\frac{1}{n}v, \frac{1}{n}v\right].
\]
We pick $m = n$. Suppose $a \in \left[-\frac{1}{n}u,\frac{1}{n}u\right]$. Then
\begin{align*}
-\frac{1}{n}u \leq &a \leq \frac{1}{n}u & \Rightarrow \\
h\left(-\frac{1}{n}u\right) \leq &h(a) \leq h\left(\frac{1}{n}u\right) & \Rightarrow \\
-\frac{1}{n}h(u) \leq &h(a) \leq \frac{1}{n}h(u) & \Rightarrow \\
-\frac{1}{n}v \leq &h(a) \leq \frac{1}{n}v & \Leftrightarrow \\
h(a) &\in \left[-\frac{1}{n}v, \frac{1}{n}v \right].
\end{align*}
This establishes the continuity of $h$. This is now enough to show that $h$ is $\R$-linear. Let $\alpha \in \R$, with $(\alpha_i)_{i \in \N}$ a sequence of rationals such that $\lim_{i \to \infty} \alpha_i = \alpha$. We have
\begin{align*}
h(\alpha a) &= h \left( \lim_{i \to \infty}\alpha_i a \right) & \text{ continuity of scalar multiplication} \\
 &= \lim_{i \to \infty} h(\alpha_i a) & h \text{ continuous} \\
 &= \lim_{i \to \infty} \alpha_i h(a) & h \text{ $\Q$-linear} \\
 &= \alpha h(a) & \text{ continuity of scalar multiplication.}
\end{align*}
So we have a map $h \in \OUS(A,B)$ such that $[0,1]_h = f$, proving that $[0,1]_\blank$ is full. 
\end{proof}

\section{The Dual Adjunction and Reflexivity}
Now we relate section \ref{BNSOUSSection} back to the adjunction in section \ref{OrigAdjSection}. We do this by showing how $\Eff$ and $\Stat$ can be defined in terms of order-unit and base-norm spaces. 

Given an effect algebra $A$, we define the \emph{signed state space} $\StatS(A)$ as follows. Let $\R^A$ be the vector space of real-valued functions $A \rightarrow \R$, let 
\[
\Stat_+(A) = \{ \phi \in \R^A \mid \exists \alpha \in \Rgeq, \psi \in \Stat(A) . \phi = \alpha \cdot \psi \},
\]
and $\StatS(A)$ be the $\R$-linear span of $\Stat_+(A)$, and define $\tau : \StatS(A) \rightarrow \R$ by $\tau(\phi) = \phi(1)$. If $f : A \rightarrow B$ is an effect algebra homomorphism, we define $\StatS(f)(\phi) = \phi \circ f$. 

\begin{lemma}
\label{StatSDefLemma}
For each effect algebra $A$, $(\StatS(A), \Stat_+(A), \tau)$ is a Banach base-norm space with base $\Stat(A)$. With the above definition on morphisms, $\StatS$ is a functor $\EA\op \rightarrow \BBNS$. The space $\StatS(A)$ admits a locally convex topology in which each element of $A$ defines a continuous functional and $\Stat(A)$ is compact, and such that $\StatS(f)$ is continuous for each effect algebra homomorphism $f : A \rightarrow B$. 
\end{lemma}
\begin{proof}
It is easy to see from the definition that $(\StatS(A), \Stat_+(A))$ is a directed ordered vector space, $\tau$ is linear and positive with base $\Stat(A)$, and that if $\Stat(A) = \emptyset$ then $\StatS(A) = \{0\}$. If $\Stat(A) = \emptyset$, there is nothing further to show, so we assume that $\Stat(A)$ is not empty for the rest of the proof. To prove that $\StatS(A)$ has radially bounded unit ball and closed positive cone, we will use the locally convex topology mentioned in the second part of the lemma. We topologize $\R^A$ with the product topology, which is locally convex \cite[II.5.2, p. 52]{schaefer}, and therefore $\StatS(A)$, equipped with the subspace topology, is locally convex. As the projection mappings from a product are continuous, each element of $A$ defines a continuous function, which is also linear. We also have that $\StatS(f)$ is continuous for this topology. The subspace $[0,1]^A \subseteq \R^A$ is compact by Tychonoff's theorem, so to show that $\Stat(A)$ is compact, we only need to show that it is closed in $[0,1]^A$. Observe that the graph of $+$ as a subset of $[0,1]^3$ is closed, because $+$ is a continuous function. Therefore, for fixed $a, b \in A$ such that $a \ovee b$ exists in $A$, the set
\[
C_{a,b} = \{ \phi : A \rightarrow [0,1] \mid \phi(a \ovee b) = \phi(a) + \phi(b) \}
\]
is a closed subset of $[0,1]^A$. Additionally, the set $C_1 = \{ \phi \in [0,1]^A \mid \phi(1) = 1 \}$ is also closed. So $\Stat(A) = C_1 \cap \bigcap\limits_{a \perp b \in A} C_{a,b}$ is closed in $[0,1]^A$. 

As $\Stat(A)$ is compact, its absolutely convex hull, being the image of $[0,1] \times \Stat(A) \times \Stat(A)$ under a continuous map, is also compact. This proves that the unit ball of $\StatS(A)$ is radially compact, and therefore radially bounded with a closed positive cone \cite[Proposition 2.2.6]{furberthesis}. As $\Stat(A)$ is compact, it admits only one uniformity, in which it is complete. This (nontrivially) implies that $\StatS(A)$ is complete in its norm topology, and so is a Banach base-norm space \cite[Proposition 2.4.13]{furberthesis}. 
\end{proof}

The existence of the compact topology above was also used in the proof that every compact convex subset of a locally convex space arises as the state space of an orthomodular lattice \cite{shultz74} and for effect algebras in \cite[Proposition 3.4.2]{roumenthesis}.

The corresponding fact for the effects on a convex set is known, originally having been proved by Ozawa \cite[\S 3]{ozawa80} in the more general case of a convex prestructure (a set $X$ equipped with a map $[0,1] \times X \times X$ not required to satisfy any further axioms). For an abstract convex set $X \in \EM(\D)$, we define $\EffS(X)$ to be the bounded affine maps from $X$ to $\R$, made into a vector space with pointwise operations. The positive cone is defined to be affine maps taking values in $\Rgeq$, and the order unit to be the element $1$ that is constantly equal to $1$ on all elements of $X$. On maps $f : X \rightarrow Y$, $\EffS(f)(b) = b \circ f$, for $a \in \EffS(Y)$. 

\begin{lemma}
\label{SignedEffectDefLemma}
For each abstract convex set $X \in \EM(\D)$, $\EffS(X)$ is a Banach order-unit space such that $[0,1]_{\EffS(X)} = \Eff(X)$. With the above definition on morphisms, $\EffS$ is a functor $\EM(\D) \rightarrow \BOUS\op$. The space $\EffS(X)$ admits a locally convex topology in which each element of $X$ defines a continuous linear functional and $\Eff(X)$ is compact, such that $\EffS(f)$ is continuous for each affine map $f : X \rightarrow Y$. 
\end{lemma}
The proof is omitted, so see either \cite[\S 3]{ozawa80}, or \cite[Prop. 2.4.15, Thm. 2.4.16]{furberthesis}, where $\EffS$ is called $\BAff$. 

We have shown how to factorize the functor $\Stat$ into $B \circ \StatS$ and the functor $\Eff$ into $[0,1]_\blank \circ \EffS$. This is enough to show that for $\eta_X$ to be an isomorphism, $X$ must be isomorphic to the base of a Banach base-norm space, and for $\epsilon_A$ to be an isomorphism, $A$ must be isomorphic to the unit interval of an order-unit space. To get a precise characterization, we must go a bit further.

As base-norm spaces $(E,E_+,\tau)$ are normed spaces, we can take their dual spaces $E^*$, \emph{i.e.} the Banach space of continuous linear maps $E \rightarrow \R$, equipped with the dual norm. Each $\phi : E \rightarrow \R$ can be restricted to the base to define $\rho_E(\phi) \in \EffS(B(E))$. Similarly, if $(A,A_+,u)$ is an order-unit space, we can restrict each continuous $a : A \rightarrow \R$ to $[0,1]_A$, to obtain $\xi_A(a) \in \StatS([0,1]_A)$. 

The dual space of a base-norm space is always a Banach order-unit space, for the following reason.

\begin{lemma}
\label{RestrictionLemma}
The restriction map $\rho$ is a natural isomorphism $\blank^* \Rightarrow \EffS \circ B$, and the restriction map $\xi$ is a natural isomorphism $\blank^* \Rightarrow \StatS \circ [0,1]_\blank$. 
\end{lemma}
\begin{proof}
For $\rho$, this is a standard fact, see \cite[Proposition 1.11]{alfsen01}, \cite[Theorem 1]{ozawa80} or \cite[Prop. 2.4.17 and Thm. 2.4.18]{furberthesis}. 

If $(A,A_+,u)$ is an order-unit space, we must first show that each $\phi \in A^*$ defines an element of $\StatS([0,1]_A)$ by restriction. As $A^*$ is a base-norm space (\cite{ellis64} or \cite[Theorem 1.19]{alfsen01}), every element of $A^*$ is in the span of the state space, so when restricted to $[0,1]_A$ each element of $A^*$ is in $\StatS([0,1]_A)$. In the other direction, every element of $\Stat([0,1]_A)$, \emph{i.e.} $\EA([0,1]_A,[0,1])$ is the restriction of a unique map of order-unit spaces $A \rightarrow \R$ by Proposition \ref{OUSFullnessProp}. As $\StatS([0,1]_A)$ is the linear span of $\Stat([0,1]_A)$, this shows that $\xi_A$ is a bijection. The fact that it is an isomorphism of base-norm spaces then follows directly from the way the base-norm space structures are defined on $A^*$ and $\StatS([0,1]_A)$, and the naturality from the fact that both functors are just precomposition for maps. 
\end{proof}

Recall that a normed space $E$ is called \emph{reflexive} if the natural embedding into its double dual $E^{**}$ is an isomorphism. Reflexive spaces, being dual spaces, are necessarily Banach spaces. We say that a base-norm space or an order-unit space is \emph{reflexive} if its underlying normed space is reflexive. We use $\RBNS$ for the full subcategory of $\BNS$ on reflexive base-norm spaces, and $\ROUS$ for the full subcategory of $\OUS$ on reflexive order-unit spaces. For an order-unit space $(A,A_+,u)$ we define $\epsilon_A' : A \rightarrow A^{**}$ as $\epsilon_A'(a)(\phi) = \phi(a)$, where $a \in A$ and $\phi \in A^*$, and for a base-norm space $(E,E_+,\tau)$ we define $\eta_E' : E \rightarrow E^{**}$ as $\eta_E'(x)(a) = a(x)$ where $x \in E$ and $a \in E^*$. 

\begin{theorem}
\label{ReflexiveCharacterizationTheorem}
For all base-norm spaces $E$, $B(\xi_{E^*}) \circ B(\eta'_E) = \Stat([0,1]_{\rho_E}) \circ \eta_{B(E)}$ and for all order-unit spaces $A$, $[0,1]_{\rho_{A^*}} \circ [0,1]_{\epsilon'_A} = \Eff(B(\xi_A)) \circ \epsilon_{[0,1]_A}$. Therefore $\eta_{B(E)}$ is an isomorphism iff $\eta'_E$ is, and $\epsilon_{[0,1]_A}$ is an isomorphism iff $\epsilon_A'$ is. As $\eta'_E$ is an isomorphism iff $E$ is reflexive, and $\epsilon'_A$ is an isomorphism iff $A$ is reflexive, $\RBNS$ is equivalent, under $B$, to the subcategory of $\EM(\D)$ where $\eta$ is an isomorphism, and $\ROUS$ is equivalent, under $[0,1]_\blank$, to the subcategory of $\EA$ where $\epsilon$ is an isomorphism. 
\end{theorem}
\begin{proof}
We first show that $B(\xi_{E^*}) \circ B(\eta_E') = \Stat([0,1]_{\rho_E}) \circ \eta_{B(E)}$. We can see that the types are correct using Lemmas \ref{StatSDefLemma} and \ref{SignedEffectDefLemma}. If we let $x \in B(E)$ and $a \in [0,1]_{E^*}$, then
\begin{align*}
(B(\xi_{E^*}) \circ B(\eta_E'))(x)(a) &= \xi_{E^*}(\eta_E'(x))(a) \\
 &= \eta_E'(x)(a) \\
 &= a(x) \\
 &= \eta_{B(E)}(x)(a) \\
 &= \eta_{B(E)}(x)([0,1]_{\rho_E}(a)) \\
 &= \Stat([0,1]_{\rho_E})(\eta_{B(E)}(x))(a).
\end{align*}
As $\xi_{E^*}$ and $\rho_E$ are isomorphisms, we have that $B(\eta_E')$ is an isomorphism iff $\eta_{B(E)}$ is. By Proposition \ref{ConvFullnessProp}, $\eta_E'$ is an isomorphism iff $B(\eta_E')$ is. Therefore $B$, when restricted to $\RBNS$, is an equivalence of $\RBNS$ with the full subcategory of objects in $\EM(\D)$ for which $\eta$ is an isomorphism. 

Now we show that $[0,1]_{\rho_{A^*}} \circ [0,1]_{\epsilon_A'} = \Eff(B(\xi_A)) \circ \epsilon_{[0,1]_A}$. We can see that the types are correct by Lemmas \ref{StatSDefLemma} and \ref{SignedEffectDefLemma}. If we let $a \in [0,1]_A$ and $\phi \in B(A^*)$, then
\begin{align*}
([0,1]_{\rho_{A^*}} \circ [0,1]_{\epsilon_A'})(a)(\phi) &= \rho_{A^*}(\epsilon_A'(a))(\phi) \\
 &= \epsilon_A'(a)(\phi) \\
 &= \phi(a) \\
 &= \epsilon_{[0,1]_A}(a)(\phi) \\
 &= \epsilon_{[0,1]_A}(a)(B(\xi_A)(\phi)) \\
 &= \Eff(B(\xi_A))(\epsilon_{[0,1]_A}(a))(\phi).
\end{align*}
We now have that $[0,1]_{\epsilon_A'}$ is an isomorphism iff $\epsilon_{[0,1]_A}$ is an isomorphism, and the rest of the theorem is proved by using Proposition \ref{OUSFullnessProp}. 
\end{proof}

Every finite-dimensional normed space is reflexive, so the self-adjoint part of every finite-dimensional C$^*$-algebra is a reflexive order-unit space, and its dual space is a reflexive base-norm space. For ease of reference, we say that an abstract convex set $X$ is reflexive iff $\eta_X$ is an isomorphism and an effect algebra is reflexive iff $\epsilon_A$ is an isomorphism. Then a convex subset $X$ of $\R^n$ is reflexive iff it is closed and bounded, or equivalently, because $\R^n$ is finite-dimensional, iff $X$ is compact. If it happens to be the case that $\Eff(\Stat(A))$ is reflexive, for $A$ an effect algebra, then $\Eff(\Stat(A))$ can be considered to be the \emph{free} reflexive effect algebra on $A$, where $\epsilon_A$ is the universal arrow. One example where this occurs is if $\Hil$ is a finite-dimensional Hilbert space of dimension $\geq 3$, then $[0,1]_{B(\Hil)}$ is the free reflexive effect algebra on $\Proj(\Hil)$, the orthomodular lattice of projections on $\Hil$, by Gleason's theorem. 

There are also infinite-dimensional reflexive base-norm and order-unit spaces. One example is the base-norm space with base the unit ball of an infinite dimensional Hilbert space $\Hil$. This base-norm space is the dual space of an infinite dimensional spin factor \cite[\S 6]{olsen}, an infinite dimensional version of a type of simple Jordan algebra occurring in Jordan, von Neumann and Wigner's classification \cite{jordan}. 

Unfortunately, C$^*$-algebras are reflexive only if they are finite-dimensional. This also means that there is no free reflexive effect algebra on the effects of an infinite-dimensional Hilbert space -- $\epsilon$ is never an isomorphism no matter how many times we take the double dual. In the next section we describe a way of fixing this by breaking the symmetry between states and effects using topologies. 

\section{Extension to Smith Spaces}
\label{SmithSpaceSection}
The following section follows along the lines of the previous section and Sections  3.3 and 3.4 of \cite{furberthesis}, which the reader is invited to consult for more details, if needed. Smith spaces are a way of making every Banach space ``reflexive'' -- instead of equipping $E^*$ with the dual norm, we use the ``bounded weak-* topology'', a modification of the weak-* topology that we shall describe in a moment. Then the double-dual embedding from $E$ into the continuous dual of $E^*$ is an isomorphism if $E$ is a Banach space, and a dense embedding for an incomplete normed space. Smith spaces can be defined abstractly as normed spaces equipped with a locally convex topology $\Tee$ in which the unit ball is compact and such that $\Tee$ is the finest topology agreeing with itself on all norm-bounded sets. For each locally convex topology $\Tee$ that makes the unit ball of a normed space $E$ compact, there is a ``Smithification'' $\Tee_b$, which is the finest topology agreeing with $\Tee$ on the norm-bounded sets of $E$, and is a Smith space topology \cite[Proposition 3.2.9]{furberthesis}. If $E$ is the dual of a Banach space, then applying this process to the weak-* topology on $E$ gives the bounded weak-* topology. 

If $E$ is a normed space, then we write $E^\sigma$ for its dual space with the bounded weak-* topology. Unlike the weak-* topology, if $\overline{E}$ is the completion of $E$, the restriction mapping $\overline{E}^\sigma \rightarrow E^\sigma$ is a linear homeomorphism. If $F$ is a Smith space, we write $F^\beta$ for the dual space of $F$ equipped with the dual norm. This is always a Banach space and has the ``strong dual topology'' \cite[Propositions 3.2.13 and 3.2.14]{furberthesis}. We can consider $\blank^\beta : \Smith \rightarrow \Norm\op$ and $\blank^\sigma : \Norm\op \rightarrow \Smith$ as functors, and then $\blank^\beta$ is a left adjoint to $\blank^\sigma$ \cite[Theorem 3.2.22]{furberthesis}, and this is an adjoint equivalence when $\blank^\sigma$ is restricted to Banach spaces \cite[Corollary 3.2.23]{furberthesis}. 

A Smith-base norm space is a quadruple $(E,\Tee,E_+,\tau)$ where $(E,E_+,\tau)$ is a base-norm space, $\Tee$ is a locally convex topology on $E$ such that $\tau$ is continuous, the base $B(E)$ is compact and such that $(E,\Tee)$ is a Smith space with respect to the norm defined by $U(E)$ \cite[\S 3.3]{furberthesis}. A Smith order-unit space is a quadruple $(A,\Tee,A_+,u)$ such that $(A,A_+,u)$ is an order-unit space, $\Tee$ a locally convex topology such that $[0,u]$ is compact, and $(A,\Tee)$ a Smith space with respect to the norm defined by $[-u,u]$ \cite[\S 3.4]{furberthesis}. The corresponding categories are $\SBNS$ and $\SOUS$, where the maps are required to be continuous with respect to the Smith space topologies. 

The following theorem extends Ozawa's \cite[Theorems 1 and 3]{ozawa80} to a categorical duality. Ozawa appears to have been the first to see the usefulness of the \emph{bounded} weak-* topology on $\EffS(X)$ instead of the weak-* topology.

\begin{theorem}
If we use the topology from Lemma \ref{SignedEffectDefLemma} to define a bounded weak-* topology on $\EffS(X)$, then $\EffS(X)$ is a Smith order-unit space, and the dual space of $\EffS(X)$ is the ``free''\footnote{Note that $X$ need not embed injectively in $\EffS(X)^*$.} Banach base-norm space on $X$. We can define a category $\CEMod$ of compact effect modules characterizing the unit intervals of $\SOUS$, and restrict the functor $\Stat$ to $\CStat : \CEMod\op \rightarrow \EM(\D)$, taking only the continuous states. Then $(\Eff,\CStat,\eta,\epsilon)$ is an adjunction where $\epsilon$ is always an isomorphism and $\eta_X$ is an isomorphism iff $X$ is isomorphic to the base of a Banach base-norm space. 
\end{theorem}
\begin{proof}
In \cite[Propositions 3.4.8 and 3.4.9]{furberthesis} it is proved that $\EffS(X)$ is a Smith order-unit space (using the notation $\BAff$ for $\EffS$). We can define a compact effect module to be an effect module\footnote{Also known as \emph{convex effect algebra}, see \cite[\S 1.2.1]{furberthesis}.} that can be embedded in a locally convex topological vector space, and the category $\CEMod$ has these as objects and continuous effect module maps as morphisms\footnote{By the fullness part of Proposition \ref{OUSFullnessProp} it does not matter whether we use effect module or effect algebra maps here, as by definition objects of $\CEMod$ are embeddable in order-unit spaces.}. More intrinsic characterizations of $\CEMod$ are given in \cite[\S 4.4]{furberthesis}, and it is proved that the unit interval functor $[0,1]_\blank : \SOUS \rightarrow \CEMod$ is an equivalence in \cite[Theorem 3.4.3]{furberthesis}. Therefore the adjunction $\EffS \dashv \CStat$ \cite[Theorem 3.4.11]{furberthesis}, in which $\epsilon$ is an isomorphism, can be composed with the equivalence $[0,1]_\blank$ to obtain the adjunction $(\Eff, \CStat, \eta, \epsilon)$, where $\Eff \dashv \CStat$, and $\epsilon$ is an isomorphism. This is also what shows that the dual space of $\EffS(X)$ is the free Banach base-norm space on $X$. If $\eta_X$ is an isomorphism, then it constitutes an embedding of $X$ as the base of the Banach base-norm space $\EffS(X)^\beta$. In the other direction, if $(E,E_+,\tau)$ is a Banach base-norm space and $X = B(E)$, then $\EffS(X) \cong E^\sigma$ by restricting linear functionals to the base (Lemma \ref{RestrictionLemma}), and so $\EffS(X)^\beta \cong E^{\sigma\beta}$, to which $E$ is isomorphic by the unit of the adjoint equivalence between $\BBNS$ and $\SOUS$ \cite[Theorem 3.4.5]{furberthesis}. As $\eta_X$ is equal to the $B$ functor applied to the composite of these isomorphisms, it is an isomorphism. 
\end{proof}

It may help to note that $X$ can be embedded as the base of a Banach base-norm space iff it is embeddable as a sequentially complete and bounded subset of a locally convex topological vector space \cite[Proposition 2.4.13]{furberthesis}, for instance if it is a closed bounded subset of a Banach space.

\begin{theorem}
If we use the topology from Lemma \ref{StatSDefLemma} on $\StatS(A)$ to define a bounded weak-* topology, then $\StatS(A)$ is a Smith base-norm space, and the dual space of $\StatS(A)$ is the ``free'' Banach order-unit space on an effect algebra. We can define a category $\CCL$ of compact convex sets, equivalent to $\EM(\Rdn)$ where $\Rdn$ is the Radon measures monad, characterizing the bases of Smith base-norm spaces, and restrict the functor $\Eff$ to $\CEff : \CCL \rightarrow \EA\op$, taking only continuous effects. Then $(\CEff, \Stat, \eta, \epsilon)$ is an adjunction where $\eta$ is always an isomorphism and $\epsilon_A$ is an isomorphism iff $A$ is isomorphic to the unit interval of a Banach order-unit space.
\end{theorem}
\begin{proof}
By \cite[Proposition 3.3.2]{furberthesis}, the Smithification of $\StatS(A)$ equipped with the topology given in Lemma \ref{StatSDefLemma} (considering $\StatS(A)$ as a subspace of $\R^A$ with the product topology) is a Smith base-norm space. If $f : A \rightarrow B$ is an effect algebra homomorphism, we show $\StatS(f)$ is continuous as follows. Consider the subbasic $\R^A$ 0-neighbourhood $N_{a,\epsilon}$ defined by an element $a \in A$, $N_a = \{ \phi \in \R^A \mid |\phi(a)| \leq \epsilon \}$. Then, with the analogous $\R^A$ $0$-neighbourhood $M_{f(a)} = \{ \phi \in \R^B \mid |\phi(f(a))| \leq \epsilon \}$, we have $\phi \in M_{f(a)}$ implies $\StatS(f)(\phi) \in N_a$, so $\StatS(f)$ is continuous with respect to the product topologies. Therefore $\StatS(f)$ is continuous with the Smith space topology on $\StatS(B)$ and the product topology on $\StatS(A)$, so by \cite[Corollary 3.2.16]{furberthesis} it is continuous when both spaces have the Smith space topologies. Combined with what was proved in Lemma \ref{StatSDefLemma}, this shows that $\StatS : \EA\op \rightarrow \SBNS$ is a functor.

The category $\CCL$ is defined to have pairs $(E,X)$ as objects, where $E$ is a locally convex topological vector space and $X \subseteq E$ is a convex set that is compact with respect to the subspace topology. A morphism $f : (E,X) \rightarrow (F,Y)$ is simply an affine continuous map $X \rightarrow Y$, with no requirement of a definition on, or extension to, $E$. The functor $B : \SBNS \rightarrow \CCL$ is an equivalence \cite[Proposition 3.3.3]{furberthesis}, so every element of $\CCL$ can be put into a canonical form ($(E,B(E))$ where $E$ is a Smith base-norm space) such that affine continuous maps extend to $E$. The forgetful functor $U : \CCL \rightarrow \CHaus$, where $\CHaus$ is the category of compact Hausdorff spaces, is monadic and the monad is isomorphic to $\Rdn$, so $\CCL \simeq \EM(\Rdn)$ \cite[Theorems 4.2.8 and 4.2.9]{furberthesis}. We can compose $B \circ \StatS$ to obtain a functor $\Stat : \EA\op \rightarrow \CCL$. 

The functor $\CEffS : \CCL \rightarrow BOUS\op$ is defined on objects $(E,X)$ to be those elements of $\EffS$ that are continuous with respect to the topology on $X$. As composition with continuous functions preserves continuity, $\CEffS$ is a functor, the restriction of $\EffS$. Then $\CEff = [0,1]_\blank \circ \CEffS : \CCL \rightarrow \EA$. It is also useful to note that the dual space functor $\blank^\beta : \SBNS \rightarrow \BOUS\op$ is isomorphic to $\CEffS \circ B$, by restricting linear functionals to the base \cite[Theorem 3.3.6]{furberthesis}.

Then $\eta$ and $\epsilon$ can be defined as in the original adjunction, satisfying the same diagrams, so $\CEff \dashv \Stat$. As $B : \SBNS \rightarrow \CCL$ is an equivalence, $\CEff \dashv \StatS$, and as $\blank^\beta : \SBNS \rightarrow \BOUS\op$ is an equivalence \cite[Theorem 3.3.7]{furberthesis}, $\StatS(A)^\beta$ is the free order-unit space on the effect algebra $A$. The map $\eta_X : X \rightarrow \Stat(\CEff(X))$ is always an isomorphism because it isomorphic as an arrow to $B(\eta_E) : B(E) \rightarrow B(E^{\beta\sigma})$, where $E$ is a Smith base-norm space, and $\eta_E$ the unit for the adjoint equivalence between $\SBNS$ and $\BOUS\op$ \cite[Theorem 3.3.7]{furberthesis}. 

If $\epsilon_A : A \rightarrow \CEff(\Stat(A))$ is an isomorphism, then $A$ has been embedded\footnote{By Proposition \ref{OUSFullnessProp}, it does not matter if we consider this as an embedding with effect algebra morphisms or convex effect algebra morphisms.} as the unit interval of the Banach order-unit space $\CEffS(\Stat(A))$. For the other direction, if $A = [0,1]_B$, where $B$ is a Banach order-unit space, then $\epsilon_A = [0,1]_{\epsilon_B}$, where $\epsilon_B$ is the counit of the adjoint equivalence $\BOUS\op \simeq \SBNS$ \cite[Theorem 3.3.7]{furberthesis} and is therefore an isomorphism. 
\end{proof}

An example of the free Banach order-unit space on an effect algebra can be given in the case of a Boolean algebra $A$. If $X$ is the Stone space of $A$, then $C(X)$ is the self-adjoint part of a commutative C$^*$-algebra, and therefore a Banach order-unit space, and for every Banach order-unit space $B$ and effect algebra homomorphism $f : A \rightarrow [0,1]_B$, there is an extension to a positive unital map $\tilde{f} : C(X) \rightarrow B$. 

\subsubsection*{Acknowledgements}
Robert Furber has been financially supported by the Danish Council for Independent Research, Project 4181-00360.

The question of whether Proposition \ref{OUSFullnessProp} is true or not, and what happens if we apply \cite[Part 0, Proposition 4.2]{lambek} arose in discussions of certain results from my thesis with Phil Scott and Pieter Hofstra on a visit to Ottawa. I was not able to give a full answer at the time, but this paper is the result.

\bibliographystyle{eptcs}
\bibliography{eaduality}

\begin{appendices}
\section{Appendix}
\label{ArchimedeannessAppendix}
In this section, we show to what extent the archimedean property is used in Proposition \ref{OUSFullnessProp}. We first remark that, as an ordered normed space that is a dual space must have a closed positive cone, which implies archimedeanness \cite[Lemma A.5.3]{furberthesis}, the non-archimedean spaces we consider in the section can never arise as dual spaces. 

In order to consider the more general situation, we temporarily redefine order-unit space to omit the archimedean property, so for this section only, an order unit space is a triple $(A,A_+,u)$ such that $(A,A_+)$ is an ordered vector space and $u$ is a strong order unit. These form a category $\OUS$ where the morphisms are linear positive unital maps. In \cite[\S 3.1, Theorem 3]{JacobsM12b} it is shown that the functor $[0,1]_\blank : \OUS \rightarrow \EMod$ is an equivalence of categories, where $\EMod$ is the category of \emph{effect modules}, which are effect algebras equipped with the structure of a module over $[0,1]$, where the morphisms are required to preserve multiplication by elements of $[0,1]$. The objects of $\EMod$ were originally defined under the name \emph{convex effect algebras} \cite{PulmannovaG98}. We saw in Proposition \ref{OUSFullnessProp} that $[0,1]_\blank$ was full when restricted to archimedean order-unit spaces. We now describe various weakenings of archimedeanness. 

In general, the Minkowski functional $\|\blank\|_{[-u,u]}$ is only a semi-norm. We say that an order-unit space $(A,A_+,u)$ is \emph{almost archimedean} (as in \cite[1.3.7]{jameson}) if $\forall n \in \N. -\frac{1}{n}u \leq a \leq \frac{1}{n}u$ implies $a = 0$, or equivalently, if $\|\blank\|_{[-u,u]}$ is a norm. So every archimedean order-unit space is almost archimedean \cite[Lemma A.5.3]{furberthesis}. is  We define an \emph{infinitesimal} to be an element $a \in A_+$ such that $a \leq \frac{1}{n}u$ for all $n \in \N$. The property of having no non-zero infinitesimals can be considered to be a form of archimedeanness, and clearly any almost archimedean order-unit space has no non-zero infinitesimals. 

We can show that these inclusions are strict using examples in 2-dimensional space (so the difference between these has nothing to do with infinite-dimensional considerations). The first and last example are also found in \cite[\S 1.3 Examples]{jameson}. 

\begin{example}\hfill
\label{NonarchimedeanExamples}
\begin{enumerate}[(i)]
\item Define $A = \R^2$, $A_+ = \{ (x,y) \in \R^2 \mid (x \geq 0 \text{ and } y = 0) \text{ or } y > 0 \}$, and $u = (0,1)$. Then $(A,A_+,u)$ is a non-archimedean order-unit space. We have $(x,y) \leq (x',y')$ iff ($x \leq x'$ and $y= y'$) or $y < y'$. So the order is the reverse lexicographic ordering. Therefore this space is known as the \emph{lexicographic plane}. The element $(1,0)$ is a non-zero infinitesimal.
\item Define $B = \R^2$, $B_+ = \{ (x,y) \in \R^2 \mid y > 0 \text{ or } (x,y) = (0,0) \}$ and $v = (0,1)$. Then $(B,B_+,v)$ is a non-archimedean order-unit space. We have $(x,y) \leq (x',y')$ iff $(x,y) = (x',y')$ or $y < y'$. This space is not almost archimedean (the point $(1,0)$ has norm $0$) but has no non-zero infinitesimals. 
\item Define $C = \R^2$, $C_+ = \{ (x,y) \in \R^2 \mid (x,y) = (0,0) \text{ or } (x > 0 \text{ and } y > 0) \}$ and $w = (1,1)$. Then $(C,C_+,w)$ is a non-archimedean order-unit space. We have $(x,y) \leq (x',y')$ iff $(x,y) = (x',y')$ or $x < x'$ and $y < y'$. This space is non-archimedean (the point $(0,-1) \leq \frac{1}{n}w$ for all $n \in \N$, but is not $\leq 0$), but is almost archimedean. 
\end{enumerate}
\end{example}
\begin{proof}
In each case, we leave it is an exercise to the reader to show that $A_+,B_+,C_+$ are cones and $u,v,w$ strong order-units. 
\begin{itemize}
\item In $A$, $(1,0)$ is infinitesimal:

As $1 \geq 0$, $(1,0) \in A_+$. Let $n \in \N$. We want to show that $(1,0) \leq \frac{1}{n}u$, or equivalently $(n,0) \leq (1,0)$, which is equivalent to $(-n,1) \in A_+$. This follows from the fact that $1 > 0$. 

\item In $B$, $(1,0)$ has norm $0$, so $(B,B_+,v)$ is not almost archimedean:

Let $n \in \N$. We have $\frac{1}{n}(0,1) - (1,0) = (-1,\frac{1}{n}) \in B_+$, so $(1,0) \leq \frac{1}{n}v$, and similarly $\frac{1}{n}(0,1) + (1,0) = (1, \frac{1}{n}) \in B_+$. Therefore $(0,1) \in \left[-\frac{1}{n}v,\frac{1}{n}v\right]$ for all $n \in \N$, so $\|(0,1)\|_{[-v,v]} = 0$, and $(B,B_+,v)$ is not almost archimedean. 

\item $(B,B_+,v)$ has no non-zero infinitesimals:

Let $(x,y) \in B_+$ be a non-zero element, so that $y > 0$. The number $\frac{1 + y}{y} > 0$, and so $\frac{1 + y}{y}(x,y) - (0,1) = \left(\frac{1 + y}{y}x, y\right) \in B_+$ because $y > 0$, so $\frac{1 + y}{y}(x,y) \geq v$, and we can take $n = \lceil \frac{1+ y}{y} \rceil$ to show $n (x,y) \geq v$, and therefore $(x,y) \geq \frac{1}{n}v > \frac{1}{n+1}v$, so $(x,y)$ is not an infinitesimal. 

\item $(C,C_+,w)$ is almost archimedean:

Let $(x,y) \in \R^2$ and suppose that $-\frac{1}{n}w \leq (x,y) \leq \frac{1}{n} w$ for all $n \in \N$. Then $\frac{1}{n} < x, y < \frac{1}{n}$ for all $n \in \N$, so $x = y = 0$. 

\item In $C$, $(-1,0) \leq \frac{1}{n}w$ for all $n \in \N$, but $(-1,0) \not\leq 0$, so $(C,C_+,w)$ is not archimedean:

First, it is clear that $(-1,0) \not\leq 0$ because $(1,0) \not\in C_+$ because $0 \not > 0$. Now, as $\frac{1}{n} > 0$, we have $\frac{1}{n}(1,1) - (-1,0) = (1 + \frac{1}{n}, \frac{1}{n}) \in C_+$, so $(-1,0) \leq \frac{1}{n}(1,1)$. 
\end{itemize}
This shows that $A, B$ and $C$ form a set of counterexamples separating these notions, as well as showing that infinitesimals are possible. 
\end{proof}

We now have the definitions and concepts to characterize which form of archimedeanness is necessary and sufficient for automatic $\R$-linearity of group homomorphisms. 

\begin{proposition}
Let $(A,A_+,u)$ be an order-unit space. The following are equivalent:
\begin{enumerate}[(i)]
\item $(A,A_+,u)$ is almost archimedean.
\item For all order-unit spaces $(B,B_+,v)$, every positive unital group homomorphism $f : (B,B_+,v) \rightarrow (A,A_+,u)$ is $\R$-linear. 
\item For all effect modules $B$, every effect \emph{algebra} homomorphism $B \rightarrow [0,1]_A$ is an effect \emph{module} homomorphism. 
\item Every positive unital group homomorphism $f : (\R,\Rgeq, 1) \rightarrow (A,A_+,u)$ is $\R$-linear (and therefore equal to the unique map that takes $\alpha \in \R$ to $\alpha u \in A$). 
\end{enumerate}
\end{proposition}
\begin{proof}\hfill
\begin{itemize}
\item (i) $\Rightarrow$ (ii):

The proof of fullness in Proposition \ref{OUSFullnessProp} does not use any archimedeanness in showing that a positive unital group homomorphism $(B,B_+,v) \rightarrow (A,A_+,u)$ is $\Q$-linear. We also have that, for any absolutely convex absorbent set $U$, $\{\frac{1}{n}U\}_{n \in \N}$ is a neighbourhood base for $0$ with respect to the $\|\blank\|_U$-topology, even though $U$ might not be the closed unit ball of $\|\blank\|_U$. The proof of continuity of positive unital group homomorphisms goes through without any archimedeanness assumption. However, if $\|\blank\|_{[-u,u]}$ is only a semi-norm, the limits of convergent sequences are not unique\footnote{The topology defined by a semi-norm $\|\blank\|$ is Hausdorff iff $\|\blank\|$ is a norm.}. But as $(A,A_+,u)$ is almost archimedean, $\|\blank\|_{[-u,u]}$ is a norm, limits of sequences are unique, and so the proof of $\R$-linearity from continuity in Proposition \ref{OUSFullnessProp} works. 
\item (ii) $\Leftrightarrow$ (iii):
Every effect module $B$ is isomorphic, as an effect module, to the unit interval of an order-unit space $(C,C_+,v)$ (see \cite[\S 3.1, Theorem 3]{JacobsM12b} or \cite{PulmannovaG98}). We take $i : B \rightarrow [0,1]_C$ to be such an isomorphism. As the $[0,1]_\blank$ functor $\OUS \rightarrow \EMod$ is full and faithful \cite[\S 3.1 Theorem 3]{JacobsM12b}, we can reason as follows. If (ii) is true and $f : B \rightarrow [0,1]_A$ is an effect algebra homomorphism, then by the first part of the fullness proof in Proposition \ref{OUSFullnessProp}, $f \circ i^{-1}$ extends to a positive unital group homomorphism $C \rightarrow A$. By (ii), this is a linear map, so $f \circ i^{-1}$ is an effect module homomorphism, so $f \circ i^{-1} \circ i = f$ is as well, proving (iii).

If (iii) is true, and $f : (C,C_+,v) \rightarrow (A,A_+u)$ is a positive unital group homomorphism, then $[0,1]_f : [0,1]_C \rightarrow [0,1]_A$ is an effect algebra homomorphism, and therefore an effect module homomorphism by (iii), and so preserves multiplications by real numbers in $[0,1]$. If $0 \neq \alpha \in \R$, define $n = \lceil |\alpha| \rceil$, so $\frac{\alpha}{n} \in [0,1]$. Then
\[
f(\alpha x) = f\left(n \frac{\alpha}{n} x\right) = n f\left(\frac{\alpha}{n}x\right) = n \frac{\alpha}{n} f(x) = \alpha f(x),
\]
so $f$ is $\R$-linear, and (ii) is true. 
\item (ii) $\Rightarrow$ (iv): \emph{A fortiori}.
\item $\lnot$(i) $\Rightarrow$ $\lnot$(iv):

Let $\xi \in [0,1]$ be an irrational number, such as $\frac{1}{\sqrt{2}}$. By Zorn's lemma, we can extend $\{1, \xi\}$ to a maximal $\Q$-linearly independent set in $\R$. This means that we can find a $\Q$-basis $(b_i)_{i \in I}$ for $\R$ where $b_0 = 1$ and $b_1 = \xi$ (the other $b_i$ necessarily being irrational, like $\xi$). As (i) is false, $(A,A_+,u)$ is not almost archimedean, so there exists a non-zero element $a \in A$ such that $-\frac{1}{n}u \leq a \leq \frac{1}{n} u$ for all $n \in \N$. Define a $\Q$-linear map $f : \R \rightarrow A$ on basis vectors by $f(b_0) = u$, $f(b_1) = \xi u + a$, and $f(b_i) = b_i u$ for all other $i \in I$. Therefore $f$ is a $\Q$-linear unital map, and we only need to show that it is positive and not $\R$-linear. 

Suppose $\beta \in \Rgeq$ and $b_1$ never occurs in its expansion in the basis $(b_i)_{i \in I}$, \emph{i.e.} $\beta = \sum_{i \in J} \beta_i b_i$, where $J$ is a finite subset of $I$ not containing $1$. Then
\[
f(\beta) = f\left(\sum_{i \in J}\beta_ib_i\right) = \sum_{i \in J}\beta_ib_iu = \beta u \geq 0
\]
because $A_+$ is a cone and $u \in A_+$. 

If, on the other hand, $b_1$ does occur in the expansion of $\beta \geq 0$, then we have
\[
f(\beta) = f\left(\beta_1b_1 + \sum_{i \in J\setminus\{1\}}\beta_ib_i\right) = \beta_1b_1u + \beta_1a + \sum_{i \in J \setminus \{1\}} \beta_ib_i u = \beta u + \beta_1 a
\]
Now, since $-\frac{1}{n}u \leq a \leq \frac{1}{n}u$ for all $n \in \N$, we have $- \gamma u \leq a \leq \gamma u$ for all $\gamma \in \Rgt$. If $\beta_1 > 0$, then $\frac{\beta}{\beta_1} > 0$, so $-\frac{\beta}{\beta_1}u \leq a$, which implies $\beta u + \beta_1 a \in A_+$. If $\beta_1 < 0$, then $-\frac{\beta}{\beta_1} > 0$, so $a \leq -\frac{\beta}{\beta_1}u$. Multiplying through by $-\beta_1$ and rearranging, we get $\beta u + \beta_1 a \in A_+$. Therefore $f(\beta) \in A_+$ in either case, as required. 

Since we have shown that $f$ is a $\Q$-linear positive unital map, all that remains is to show that it is not $\R$-linear. We have $\xi \cdot f(1) = \xi \cdot u$, and $f(\xi \cdot 1) = \xi \cdot u + a$. These are not equal because $a \neq 0$. So (iv) is false.
\end{itemize}
Therefore (iv) $\Rightarrow$ (i) by contraposition.
\end{proof}

Note that the above shows that if an effect algebra $A$ has two effect module structures, one of which embeds as an effect module in an almost archimedean order-unit space, then the identity map $A \rightarrow A$ is an effect module isomorphism between them, so they are both the same. 

\end{appendices}
\end{document}